\newcommand*{\PAPER}{}%

\ifdefined\PAPER
\documentclass[conference,10pt,onecolumn,final,a4paper]{IEEEtran}
\fi

\usepackage{epsfig,endnotes}
\usepackage{color, colortbl}
\usepackage{tikz}
\usepackage{graphicx}
\usepackage{amssymb,amsmath}
\usepackage{booktabs}
\usepackage{verbatim}
\usepackage{multirow}
\usepackage{subcaption}
\usepackage{hyperref,cleveref}
\usepackage{url}
\usepackage[linesnumbered,ruled,vlined]{algorithm2e}
\usepackage{amsthm}
\usepackage{enumerate}

\newtheorem{prop}{Proposition}

\usepackage{times}

\hypersetup{
    colorlinks=true,
    linkcolor=black,
    citecolor=black,
    filecolor=black,
    urlcolor=black,
}

\newtheorem{theorem}{Theorem}[section]

\newtheorem{lemma}[theorem]{Lemma}

\crefname{section}{§}{§§}
\Crefname{section}{§}{§§}

\definecolor{DarkGray}{gray}{0.85}

\newcommand{\AC}[0]{\textsc{AllConcur}}

\newcommand{\tobroadcast}[1]{\emph{A}-\emph{broa\-dcast}$(\mathit{#1})$}
\newcommand{\todeliver}[1]{\emph{A}-\emph{deliver}$(\mathit{#1})$}
\newcommand{\sender}[1]{$\mathit{sender}(\mathit{#1})$}

\DeclareMathAlphabet{\mathpzc}{OT1}{pzc}{m}{it}

\newlength\mylen

\usepackage{booktabs} % For formal tables

\begin{document}

\sloppy

%don't want date printed
\date{}

\title{Formal Specification and Safety Proof of a Leaderless Concurrent Atomic Broadcast Algorithm}

\author{
\IEEEauthorblockN{Marius Poke}
\IEEEauthorblockA{HLRS\\
University of Stuttgart\\
marius.poke@hlrs.de}
\and
\IEEEauthorblockN{Colin W. Glass}
\IEEEauthorblockA{HLRS\\
University of Stuttgart\\
glass@hlrs.de}
}

\maketitle

%%%% ABSTRACT %%%%
\begin{abstract}
%%  Motivation/problem statement: Why do we care about the problem? 
%%  What practical, scientific, theoretical or artistic gap is your research filling?
Agreement plays a central role in distributed systems working on a common task.
The increasing size of modern distributed systems makes  them more
susceptible to single component failures.
Fault-tolerant distributed agreement protocols rely for the most part on leader-based atomic broadcast algorithms, 
such as Paxos. 
Such protocols are mostly used for data replication, which requires only a small number of servers to reach agreement. 
Yet, their centralized nature makes them ill-suited for distributed agreement at large scales. 
The recently introduced atomic broadcast algorithm \AC{} enables high throughput for distributed agreement 
while being completely decentralized.
% 
% Recently, \AC{} was introduced---a leaderless concurrent atomic broadcast algorithm that enables
% higher throughput for distributed agreement while being completely decentralized~\cite{poke2017allconcur}.
%%  Methods/procedure/approach: What did you actually do to get your results? 
%% (e.g. analyzed 3 novels, completed a series of 5 oil paintings, interviewed 17 students)
In this paper, we extend the work on \AC{} in two ways. 
First, we provide a formal specification of \AC{} that enables a better understanding of the algorithm. 
Second, we formally prove \AC{}'s safety property on the basis of this specification. 
%%  Conclusion/implications: What are the larger implications of your findings, especially 
%% for the problem/gap identified in step 1?
%
Therefore, our work not only ensures operators  safe usage of \AC{}, 
but also facilitates the further improvement of distributed agreement protocols based on \AC{}. 
\end{abstract}

\begin{IEEEkeywords}
Distributed Agreement; Atomic Broadcast; Design Specification; Safety Proof; TLA+
\end{IEEEkeywords}

\section{Introduction}

Distributed systems working on a common task often have a shared state. 
Typically, the ordering of updates to this shared state is relevant, 
i.e. only identical ordering guarantees identical resulting states. 
Thus, to ensure identical distributed states, distributed agreement is required. 
Both for systems handling critical data and very large systems, the ability to sustain failures is crucial. 
For the former, inconsistencies are unacceptable, for the latter, failures become so common, 
that the lack of fault tolerance becomes a performance issue. 
Atomic broadcast algorithms enable fault-tolerant distributed agreement.

In this paper we formally specify \AC{}~\cite{poke2017allconcur}, a protocol that provides distributed agreement through a leaderless 
concurrent atomic broadcast algorithm under the assumption of 
partial synchrony~(\cref{sec:allconcur}). 
\AC{}\footnote{We use \AC{} to refer to both the distributed agreement protocol and the atomic broadcast algorithm.} 
enables decentralized distributed agreement among a group of servers that 
communicate over an overlay network described by a sparse digraph. 
Moreover, it requires subquadratic work per server for every agreement instance 
and it significantly reduces the expected agreement time by employing an early termination mechanism. 
\AC{}'s specification is based on the original description of the protocol~(\cref{sec:design}). 
In addition, we provide a mechanically verifiable proof of \AC{}'s safety~(\cref{sec:proof}); the proof follows the steps 
of the informal proof described in the original paper~\cite{poke2017allconcur}. 

\subsection{Related work and motivation}
% atomic broadcast 
Atomic broadcast plays a central role in fault-tolerant distributed systems; 
for instance, it enables the implementation of both state machine replication~\cite{Schneider:1990:IFS:98163.98167,Lamport1978} 
and distributed agreement~\cite{poke2017allconcur,Unterbrunner2014}.
As a result, the atomic broadcast problem sparked numerous proposals for algorithms~\cite{Defago:2004:TOB:1041680.1041682}.
Many of the proposed algorithms rely on a distinguished server (i.e., a leader) to provide total order;
yet, the leader may become a bottleneck, especially at large scale.
% destinations agreement -> consensus 
As an alternative, total order can be achieved by
destinations agreement~\cite{Defago:2004:TOB:1041680.1041682,Chandra:1996:UFD:226643.226647,poke2017allconcur}. 
On the one hand, destinations agreement enables decentralized atomic broadcast algorithms; 
on the other hand, it entails agreement on the set of delivered messages and, thus, it requires consensus.

% leader-based consensus - single-leader
Most consensus algorithms and implementations are 
leader-based~\cite{Lamport:1998:PP:279227.279229, Lamport2001, Liskov2012, Ongaro2014, 
Burrows:2006:CLS:1298455.1298487, Poke:2015:DHS:2749246.2749267,Chandra:2007:PML:1281100.1281103,Corbett:2012:SGG:2387880.2387905};
thus, one server is on the critical path for all communication. 
% leader-based consensus - multi-leader
Several attempts were made to increase performance by adopting 
a multi-leader approach~\cite{Moraru:2013:MCE:2517349.2517350, Losa:2016:BAF:2933057.2933072,Mao:2008:MBE:1855741.1855767}.
Still, such approaches assume the overlay network is described by a complete 
digraph, i.e., each server can send messages to any other server.
In general, leader-based consensus algorithms are intended for data replication, where the number of replicas (i.e., servers) are bounded 
by the required level of data reliability~\cite{Poke:2015:DHS:2749246.2749267}.
Distributed agreement has no such bound---the number of agreeing servers is an input parameter and it can be in the range of thousands or more. 
Thus, solutions that rely (for communication) on a complete digraph are not suitable for distributed agreement.

% motivation 
\AC{} provides leaderless consensus and, thus, leaderless atomic broadcast, 
by using (as overlay network) any digraph with a vertex-connectivity exceeding the maximum number of tolerated failures. 
The original paper provides both a detailed description of the \AC{} algorithm and an informal proof of its correctness~\cite{poke2017allconcur}.
Yet, a formal specification allows for a better understanding of the algorithm and provides the basis for a formal proof of correctness. 
To formally specify \AC{}'s design, we use the TLA+ specification language~\cite{Lamport:2002:SST:579617}. 
Then, we use the TLA+ Proof System (TLAPS)~\cite{tlaps} to formally proof that \AC{}'s specification guarantees safety. 
Both the specification and the mechanically verifiable proof are publicly available~\cite{allconcur_tla}.

% wrap up
\textbf{Key contributions.}
In summary, our work makes two key contributions:
\begin{itemize}
\setlength{\itemsep}{0pt}
  \item a formal specification of \AC{}---a distributed agreement protocol that 
  relies on a leaderless concurrent atomic broadcast algorithm~(\cref{sec:design});
  \item a formal proof of \AC{}'s safety property~(\cref{sec:proof}).
\end{itemize}

\section{Overview}
\label{sec:overview}

This section describes the system model we consider for solving the atomic broadcast problem~(\cref{sec:abcast}) 
and it provides an overview of \AC{} and its early termination mechanism~(\cref{sec:allconcur}). 

\subsection{The atomic broadcast problem}
\label{sec:abcast}

We consider $n$ servers that are subject to a maximum of $f$ fail-stop failures.
The servers communicate via messages according to an overlay network, 
described by a digraph $G$---server $p$ can send a message to another 
server $q$ if $G$ has an edge $(p,q)$. 
The communication is reliable, i.e., messages cannot be lost, only delayed; 
also, message order is preserved by both edges and nodes.
We use the terms node and server interchangeably. 

To formaly define atomic broadcast, we use the notations from Chandra and Toueg~\cite{Chandra:1996:UFD:226643.226647}:
$m$ is a message (that is uniquely identified);
\tobroadcast{m} and \todeliver{m} are communication primitives for broadcasting and delivering messages atomically;
and \sender{m} denotes the server that A-broadcasts~$m$.
Then, any (non-uniform) atomic broadcast algorithm must satisfy four properties~\cite{Chandra:1996:UFD:226643.226647, Hadzilacos:1994:MAF:866693}:
\begin{itemize}
\setlength{\itemsep}{0pt}
  \item (Validity) If a non-faulty server A-broadcasts $m$, then it eventually A-delivers $m$.
  \item (Agreement) If a non-faulty server A-delivers $m$, then all non-faulty servers eventually 
  A-deliver $m$.
  \item (Integrity) For any message $m$, every non-faulty server A-delivers $m$ at most once, and only if 
  $m$ was previously A-broadcast by \sender{m}.
  \item (Total order) If two non-faulty servers $p$ and $q$ A-deliver messages $m_1$ and $m_2$, 
  then $p$ A-delivers $m_1$ before $m_2$, if and only if $q$ A-delivers $m_1$ before $m_2$. 
\end{itemize}
Integrity and total order are safety properties; validity and agreement are liveness property~(\cref{sec:safe_and_live}).
If validity, agreement and integrity hold, the broadcast is \emph{reliable}~\cite{Hadzilacos:1994:MAF:866693, Chandra:1996:UFD:226643.226647}.
We consider reliable broadcast algorithms that use (as overlay networks) digraphs with the vertex-connectivity 
exceeding $f$; thus, despite any $f$ failures, the servers remain connected.
Moreover, we consider atomic broadcast algorithms that provide total order through
\emph{destinations agreement}~\cite{Defago:2004:TOB:1041680.1041682}---all 
non-faulty servers reach \emph{consensus} on the set of messages to be 
A-delivered in a deterministic order.

% lower bound -> early termination 
In a synchronous round-based model~\cite[Chapter~2]{attiya2004distributed}, 
consensus requires (in the worst case) at least $f+1$ rounds~\cite{Aguilera:1998:SBP:866987}.
Clearly, if $G$ is used as overlay network, consensus requires (in the worst case) $f+D_f(G,f)$ rounds, 
where $D_f(G,f)$ denotes $G$'s 
\emph{fault diameter}---$G$'s diameter after removing any $f$ nodes~\cite{Krishnamoorthy:1987:FDI:35064.36256}.
Yet, always assuming the worst case is inefficient: It is very unlikely for the number of rounds to exceed $D_f(G,f)$, 
if the mean time between failures is long compared to the length of rounds~\cite{poke2017allconcur}. 
%Thus, we consider algorithms that avoid assuming always the worst case.

\subsection{\AC{}}
\label{sec:allconcur}
We consider \AC{}~\cite{poke2017allconcur}---a leaderless concurrent atomic broadcast algorithm 
that adopts a novel early termination mechanism to avoid assuming always the worst case.
\AC{} is round-based: In every round, every server A-broadcasts a (possibly empty) message and 
then A-delivers all known messages in a deterministic order. 
Since a synchronous model is impractical, while an asynchronous model makes solving consensus impossible~\cite{Fischer:1985:IDC:3149.214121},  
\AC{} assumes the following model of partial synchrony---message delays can be approximated by a known distribution~\cite{poke2017allconcur}. 
Under this assumption, a heartbeat-based failure detector (FD)~\cite{Chandra:1996:UFD:226643.226647}---every server 
sends heartbeats to its successors and, if it fails, the successors detects the lack of heartbeats---can be treated 
(with a certain probability) as a \emph{perfect} FD. 
Specifically, completeness (i.e., all failures are eventually detected) is deterministically guaranteed, 
while accuracy (i.e., no server is falsely suspected to have failed) is probabilistically guaranteed~\cite{poke2017allconcur}.

\AC{}'s early termination mechanism uses failure notifications to track A-broadcast messages. 
Every server maintains a \emph{tracking digraph} for every A-broadcast message. 
The nodes of a tracking digraph are the servers suspected to have the A-broadcast message, 
while the edges indicate the suspicion of how the message was transmitted. 
To stop tracking an A-broadcast message, a server must either receive the message or 
suspect only failed servers to have the message. A server can safely A-deliver 
its known messages once it stops tracking all A-broadcast messages~(\cref{sec:proof}).

The early termination mechanism relies on the following proposition:
\begin{prop}
\label{prop:in_order_trans}
If a non-faullty server receives an A-broadcast message and, subsequently, reliably broadcasts a failure notification, 
any other non-faulty server receives the A-broadcast message prior to the failure notification.
\end{prop}
The reason Proposition~\ref{prop:in_order_trans} holds is twofold: 
(1) a non-faulty server sends further (to its successors) any message it receives for the first time; 
and (2) message order is preserved by both edges and nodes.
 
\begin{figure}[!tp]
\captionsetup[subfigure]{justification=centering}
\centering
\subcaptionbox{$G_S(9,3)$\label{fig:gs_n9_d3}} {
\includegraphics[width=.20\textwidth]{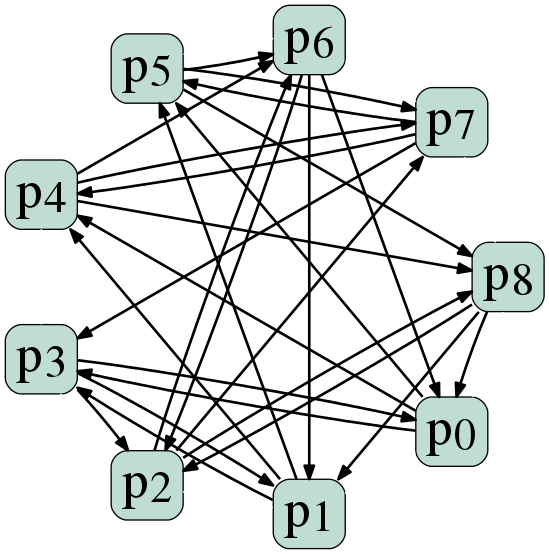}
}
\qquad
\subcaptionbox{\label{fig:tracking_digraph}} {
\scalebox{0.95}{\input{figures/tracking_digraph_small.tex}}
}
\caption{(a) Overlay network connecting nine servers. 
(b) Server $p_6$ tracking $p_0$'s message $m_0$ within a $G_S(9,3)$ digraph. Dotted red nodes indicate failed servers.
Dashed gray nodes indicate servers from which $p_6$ received failures notifications (i.e, dashed red edges). 
Solid green nodes indicate servers suspected to have $m_0$.}
\label{fig:tracking_in_gs}
\end{figure}

To illustrate the early termination mechanism we consider an example similar to the one in the \AC{} paper~\cite{poke2017allconcur}:
$n=9$ servers $(p_0,\ldots,p_8)$ connected through a $G_S(n,d)$ 
digraph~\cite{Soneoka:1996:DDC:227095.227101} with $d=3$ the digraph's degree (see Figure~\ref{fig:gs_n9_d3}).
The digraph is regular and optimally connected~\cite{Meyer:1988:FFG:47054.47067}, i.e., the vertex-connectivity equals the degree; hence, $f=2$. 
Figure~\ref{fig:tracking_digraph} shows the changes to the tracking digraph used by $p_6$ 
to track $p_0$'s message $m_0$ under the following failure scenario: $p_0$ fails after sending $m_0$ only to $p_5$, 
which receives $m_0$, but fails without sending it further. 
When $p_6$ receives the first notification of $p_0$'s failure, for example originating from $p_4$, 
it adds to the digraph all of $p_0$'s successors except for $p_4$, the sender of the notification---$p_6$ suspects that, 
before $p_0$ failed, it sent $m_0$ to all its successors, except $p_4$, which could not have received $m_0$ from $p_0$.
Had $p_4$ received $m_0$ from $p_0$, then $m_0$ would have arrived to $p_6$ before the failure notification (cf.~Proposition~\ref{prop:in_order_trans}).  
Also, $p_6$ tracks $m_0$ until it only suspects failed servers to have it; 
only then is $p_6$ sure no non-faulty server has $m_0$ and stops tracking it.

\section{\AC{}: design specification}
\label{sec:design}

We use TLA+~\cite{Lamport:2002:SST:579617} to provide a formal design specification of \AC{}~\cite{allconcur_tla}.
In addition to the number of servers $n$ and the fault tolerance 
$f$~(\cref{sec:abcast}),
we define $\mathit{S}$ to be the set of servers and $E$ the set of directed edges describing the overlay network; 
clearly, $\mathit{E} \subseteq \mathit{S} \times \mathit{S}$. 
To define the digraph 
$G$ (i.e., the overlay network), 
we use the \emph{Graphs} module~\cite{Lamport:2002:SST:579617}: 
$G$ is defined as a record whose $\mathit{node}$ field is $\mathit{S}$ and $\mathit{edge}$ field is $\mathit{E}$.
We assume that $G$'s vertex-connectivity is larger than $f$.

We split the design of \AC{} into three modules: 
(1) an atomic broadcast (AB) module; (2) a networking (NET) module; and (3) a failure detector (FD) module.
Figure~\ref{fig:design} illustrates the three modules together with the variables describing \AC{}'s state;
moreover, the arrows indicate the flow of information, e.g., 
receiving a failure notification updates the set $F[p]$, which leads to the update of the tracking digraphs in $g[p]$~(\cref{sec:ab_spec}). 
The AB module is the core of \AC{}'s design~(\cref{sec:ab_spec}). 
It exposes two interfaces at every server~$p \in \mathit{S}$: 
the input interface $\mathtt{Abcast}(p)$, to A-broadcast a message; 
and the output interface $\mathtt{Adeliver}(p)$, to A-deliver all known A-broadcast messages.
The AB module relies on the other two modules for interactions between 
servers:
the NET module~(\cref{sec:net_spec}) provides an interface for asynchronous message-based communication; 
and the FD module~(\cref{sec:fd_spec}) provides information about faulty servers~\cite{Chandra:1996:UFD:226643.226647}.

\begin{figure}[!tp]
\centering
\scalebox{0.85}{\input{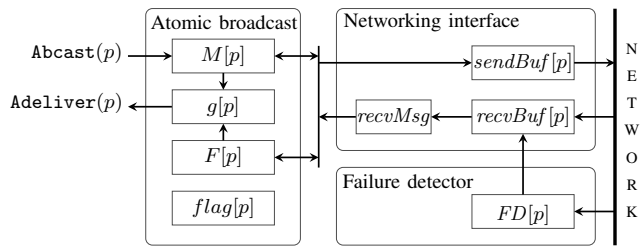}}
\caption{\AC{} system design from the perspective of a server $p$. Boxes depict the state (i.e., the variables), 
while arrows indicate the flow of information.}
\label{fig:design}
\end{figure}

\subsection{The atomic broadcast module}
\label{sec:ab_spec}

Let $p \in \mathit{S}$ be any server. 
Then, the state of the AB module is described by the values of four variables (see Figure~\ref{fig:design}):
(1)~$M[p]$ is the set of A-broadcast messages known by $p$; actually, $M[p]$ contains the messages owners, 
i.e., $M[p] \subseteq \mathit{S}$.
(2)~$g[p]$ is an array of $n$ tracking digraphs, one per server; 
the digraph $g[p][q]$ is used by $p$ to track the message A-broadcast by $q$.
(3)~$F[p]$ is an array of $n$ sets, one per server; 
the set $F[p][q]$ contains all servers from which $p$ received notifications of $q$'s failure.
(4)~$\mathit{flag}[p]$ is a record with three binary fields---$\mathit{nf}$ indicates whether $p$ is non-faulty, 
$\mathit{ab}$ indicates whether $p$ A-broadcast its message, and $\mathit{done}$ indicates whether $p$ terminated.
For ease of notation, we omit the $\mathit{flag}$ prefix and refer to the three flags as $\mathit{nf}[p]$, 
$\mathit{ab}[p]$ and $\mathit{done}[p]$, respectively.

\textbf{Initial state.}
Initially, all sets in both $M[p]$ and $F[p]$ are empty---$p$ neither knows of any A-broadcast message 
nor has received any failure notifications. For any $q \in \mathit{S}$, $g[p][q]$ contains 
only $q$ as node---$p$ suspects that each message is only known by their owner. 
Also, $p$ is initially non-faulty and it has neither A-broadcast its message nor terminated.

\begin{table*}[!tp]
\centering
% \scriptsize
\begin{tabular}{ l | c c c c c c c c c c }

  \cmidrule[1.5pt](){1-11}

   &
  $M$ & $g$ & $F$ & 
  $\mathit{nf}$ & $\mathit{ab}$ & $\mathit{done}$ &
  $\mathit{sendBuf}$ & $\mathit{recvBuf}$ & $\mathit{recvMsg}$ & $\mathit{FD}$ \\

  \hline
  
   \rowcolor{DarkGray}  
   $\mathtt{Abcast}(p)$ &
   $[p]$ & $[p][p]$ & - &   
   - & $[p]$ & - &   
   - & - & - & - \\
   
   $\mathtt{Adeliver}(p)$ &
   - & - & - &   
   - & - & $[p]$ &   
   - & - & - & - \\ 
   
   \rowcolor{DarkGray}  
   $\mathtt{RecvBCAST}(p,m)$ &
   $[p]$ & $[p][m.o]$ & - &   
   - & - & - &   
   - & - & - & - \\
   
   $\mathtt{RecvFAIL}(p,m)$ &
   - & $[p]$ & $[p][m.t]$ &   
   - & - & - &   
   - & - & - & - \\
   
   \rowcolor{DarkGray}  
   $\mathtt{Fail}(p)$ &
   - & - & - &   
   $[p]$ & - & - &   
   - & - & - & - \\
   
%    \hline
   
   $\mathtt{SendMsg}(p,\ldots)$ &
   - & - & - &   
   - & - & - &   
   $[p]$ & - & - & - \\
   
   \rowcolor{DarkGray}
   $\mathtt{TXMsg}(p)$ &
   - & - & - &   
   - & - & - &   
   $[p]$ & $[q \in p^+(G)]$ & - & - \\
   
   $\mathtt{DeliverMsg}(p)$ &
   - & - & - &   
   - & - & - &   
   - & $[p]$ & $[\ast]$ & - \\
   
%    \hline
   
   \rowcolor{DarkGray}  
   $\mathtt{DetectFail}(p,q)$ &
   - & - & - &   
   - & - & - &   
   - & $[p]$ & - & $[p][q]$ \\
   
  \cmidrule[1.5pt](){1-11} 
\end{tabular}
  \caption{The effect of the next-state relations on \AC{}'s state. 
  The brackets $[]$ indicate what elements of the variables are modified; 
  $[\ast]$ indicates that the entire variable is modified.
  Also, $p^+(G)$ denotes the set of successors of $p$ in $G$;
  $m.\mathit{o}$ denotes the server that first sent message $m$; 
  and $m.\mathit{t}$ denotes the server targeted by a failure notification $m$, i.e., the failed server.}
\label{tab:state_changes}
\end{table*}

\textbf{Next-state relations.}
The AB module defines six operators that specify all the possible state 
transitions (see Table~\ref{tab:state_changes}).
In addition to the two exposed interfaces, i.e., $\mathtt{Abcast}(p)$ and $\mathtt{Adeliver}(p)$,
$p$ can perform the following four actions: (1) receive a message, i.e., $\mathtt{ReceiveMessage}(p)$; 
(2) invoke the NET module for transmitting a 
message, i.e., $\mathtt{TXMsg}(p)$~(\cref{sec:net_spec});
(3) fail, i.e., $\mathtt{Fail}(p)$; and 
(4) invoke the FD module for detecting the failure of a 
predecessor $q \in \mathit{S}$, i.e., $\mathtt{DetectFail}(p,q)$~(\cref{sec:fd_spec}).

The $\mathtt{Abcast}(p)$ operator updates $M[p]$, $g[p]$ and $\mathit{ab}[p]$---it adds $p$ to $M[p]$; 
it removes all servers from $g[p][p].\mathit{node}$; and it sets the $\mathit{ab}[p]$ flag. 
Also, it sends $p$'s message further by invoking 
the $\mathtt{SendMsg}$ operator of the NET module~(\cref{sec:net_spec}).
The main precondition of the operator is that $p$ has not A-broadcast its message already; 
hence, a message can be A-broadcast at most once. 
%The operator is enabled only if $p$ is non-faulty and it has neither A-broadcast its message nor terminated.

The $\mathtt{Adeliver}(p)$ operator sets the $\mathit{done}[p]$ flag; 
as a result, $p$ can A-deliver the messages in $M[p]$ in a deterministic order. 
The main precondition of the operator is that all $p$'s tracking digraphs are empty,
i.e., $g[p][q].\mathit{node} = \emptyset,\,\forall q \in \mathit{S}$. 
In Section~\ref{sec:proof}, we show that this precondition is sufficient for safety. 

The $\mathtt{ReceiveMessage}(p)$ operator invokes 
the $\mathtt{DeliverMsg}$ operator of the NET module; 
as a result, the least recent message from the $\mathit{recvBuf}$ is stored into $\mathit{recvMsg}$~(\cref{sec:net_spec}). 
\AC{} distinguishes between A-broadcast messages and failure notifications. 
For both, the $o$ field indicates the server that first sent the message, i.e., the \emph{owner}; 
also, the $t$ field of a failure notification indicates the server suspected to have failed, i.e., the \emph{target}.
When $p$ receives a message $m$, it invokes one of the following operators---$\mathtt{RecvBCAST}(p, m)$
or $\mathtt{RecvFAIL}(p,m)$. 
To avoid resends, both operators are enabled only if $p$ has not already received~$m$. 

The $\mathtt{RecvBCAST}(p, m)$ operator updates both $M[p]$ and $g[p]$---it adds $m.o$ to $M[p]$; 
and it removes all servers from $g[p][m.o].\mathit{node}$. 
Also, it sends $m$ further by invoking 
the $\mathtt{SendMes\-sage}$ operator of the NET module~(\cref{sec:net_spec}).
In addition, if $p$ has not A-broadcast its message, receiving $m$ triggers the $\mathtt{Abcast}(p)$ operator.
One of the precondition of the operator is that $m$ was A-broadcast by its owner; although this condition is ensured 
by design, it facilitates the proof of the integrity property~(\cref{sec:proof}).

The $\mathtt{RecvFAIL}(p,m)$ operator updates both $F[p]$ and $g[p]$---it adds $m.o$ to $F[p][m.t]$ 
and updates every tracking digraph, in $g[p]$, that contains $m.t$. 
Updating the tracking digraphs is the core of \AC{}'s early termination mechanism and we describe it in details in Section~\ref{sec:td_update}.

The $\mathtt{Fail}(p)$ operator clears the $\mathit{nf}[p]$ flag.
As a result, all of $p$'s operators are disabled---\AC{} assumes a fail-stop model.
The main precondition of the operator is that less than $f$ servers have already failed. 

$\mathtt{TXMsg}(p)$ and $\mathtt{DetectFail}(p,q)$ are discussed in Section~\ref{sec:net_spec} and Section~\ref{sec:fd_spec}, respectively. 

\subsubsection{Updating the tracking digraphs}
\label{sec:td_update}

Tracking digraphs are trivially updated when $p$ adds an A-broadcast message to the set $M[p]$  
and, as a result, removes all the servers from the digraph used (by $p$) to track this message.
Updating the tracking digraphs becomes more involved when $p$ receives a failure notification. 
Let $p_t$ be the target and $p_o$ the owner of a received failure notification, i.e., $p_o$ detected $p_t$'s failure. 
Then, $p$ updates every tracking digraph, in $g[p]$, that contains $p_t$. Let $g[p][p_\ast]$ be such a digraph; if, after the update, $g[p][p_\ast]$ contains only servers known by $p$ to have failed, then 
it is completely pruned---$p$ is certain no non-faulty server can have $m_{\ast}$.
We specify two approaches to update $g[p][p_\ast]$. 
The first approach follows the algorithm described in the \AC{} paper~\cite{poke2017allconcur};
yet, due to its recursive specification, it is not suitable for the TLA+ Proof System (TLAPS)~\cite{tlaps}. 
The second approach constructs the tracking digraph from scratch, using the failure notifications from $F[p]$;
yet, it requires the TLA+ Model Checker~\cite{Lamport:2002:SST:579617} to enumerate all paths of a digraph.

\textbf{First approach---Recursive specification.}
Let $p_t^+(g[p][p_\ast])$ be the set of $p_t$'s successors in $g[p][p_\ast]$.
Then, if $p_t^+(g[p][p_\ast]) \neq \emptyset$, then $p$ already received another notification of $p_t$'s failure, 
which resulted in $p$ suspecting $p_t$ to have sent (before failing) $m_\ast$ to its other successors, including $p_o$; 
thus, after receiving the failure notification, $p$ removes the edge $(p_t,p_o)$ from $g[p][p_\ast]$.
Removing an edge, may disconnect some servers from the root $p_\ast$; intuitively, these servers cannot 
have received $m_\ast$ from any of the other suspected servers and thus, are removed from $g[p][p_\ast]$. 

Yet, if $p_t^+(g[p][p_\ast]) = \emptyset$ (i.e., this is the first notification of $p_t$'s failure $p$ receives),
then, we update $g[p][p_\ast]$ using a recursive function that takes two arguments---a FIFO queue $Q$ and a digraph $\mathit{td}$.
Initially, $Q$ contains all the edges (in $G$) connecting $p_t$ to its successors (except $p_o$) and $\mathit{td} = g[p][p_\ast]$.
Let $(x,y)$ be an edge extracted from $Q$. 
Then, if $y \in g[p][p_\ast]$,  (i.e., $y$ is suspected by $p$ to have $m_\ast$),
the edge $(x,y)$ is added to $\mathit{td}$ (i.e., $p$ suspects $y$ got $m_\ast$ from $x$). 
If $y \notin g[p][p_\ast]$, $y$ is also added to $\mathit{td}$ (i.e., $p$ suspects $y$ has $m_\ast$). 
In addition, if $F[p][y] \neq \emptyset$ (i.e., $y$ is known by $p$ to have failed),
the edges connecting $y$ to its successors (expect those known to have failed) 
are added to $Q$. Finally, the function is recalled with the updated~$Q$ and~$\mathit{td}$.
The recursion ends when $Q$ is empty.

\textbf{Properties of tracking digraphs.} 
From the recursive specification, we 
deduce the following four invariants that uniquely define 
a non-empty tracking digraph, denoted by $\mathtt{TD}(p, p_\ast)$:
\begin{enumerate}[($\text{I}_1$)]
\setlength{\itemsep}{0pt}
  \item it contains its root, i.e., $p_\ast \in \mathtt{TD}(p, p_\ast).\mathit{node}$;
  \item it contains all the successors of every server (in the digraph) known to have failed, 
  except those successors from which failure notifications were received, i.e., 
\begin{align*}
 &\forall \mathtt{TD}(p, p_\ast).\mathit{node} : F[p][q] \neq \emptyset \Rightarrow
  \forall q_s \in q^+(G) \setminus F[p][q] \,:\, q_s \in \mathtt{TD}(p, p_\ast).\mathit{node};
\end{align*}
  \item it contains only edges that connect a server (in the digraph) known to have failed
  to all its successors, except those successors from which failure notifications were received, i.e., 
\begin{align*}
  \forall (e_1,e_2) \in \mathtt{TD}(p, p_\ast).\mathit{edge} \,:\, (F[p][e_1] \neq \emptyset \wedge e_2 \in e_1^+(G) \setminus F[p][e_1]);
\end{align*}   
  \item it contains only servers that are either the root or 
  the successor of another server (in the digraph) known to have failed, 
  except those successors from which failure notifications were received, i.e., 
\begin{align*}
 \forall q \in \mathtt{TD}(p, p_\ast).\mathit{node} \,:\, (q = p_\ast \vee (\exists q_p \in \mathtt{TD}(p, p_\ast).\mathit{node} \,:\, 
          &\wedge F[p][q_p] \neq \emptyset \\
          &\wedge q \in q_p^+(G) \setminus F[p][q_p])).
\end{align*}   
\end{enumerate}

The intuition behind invariant $\text{I}_1$ is straightforward---while tracking $p_\ast$'s message, 
$p$ always suspects $p_\ast$ to have it.
Invariants $\text{I}_2$ and $\text{I}_3$ describe how a tracking digraph expands. 
The successors of any server $q$, which is both suspected to have $m_\ast$ and known to have failed, 
are suspected (by $p$) to have received $m_\ast$ directly from $q$ before it failed. 
Yet, there is one exception---successors from which $p$ already received notifications of $q$'s failure.
Receiving a notification of $q$'s failure before receiving $m_\ast$ entails the sender 
of the notification could not have received $m_\ast$ directly from $q$ (cf.~Proposition~\ref{prop:in_order_trans}). 

\begin{figure}
\captionsetup[subfigure]{justification=centering}
\centering
\subcaptionbox{$G_S(9,3)$\label{fig:gs_n9_d3_again}} {
\includegraphics[width=.20\textwidth]{figures/gs_n9_d3.png}
}
\qquad\qquad
\subcaptionbox{$\bar{K}_9(p_6)$\label{fig:k9_p6}} {
\includegraphics[width=.20\textwidth]{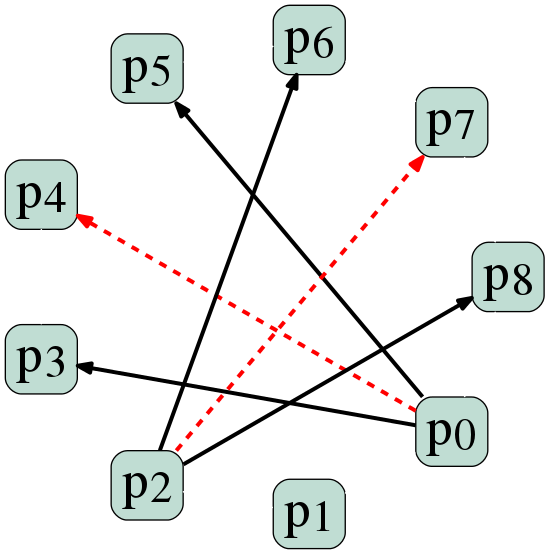}
}
\qquad\qquad
\subcaptionbox{$g[p_6][p_0]$\label{fig:g_p6_p0}} {
\includegraphics[width=.20\textwidth]{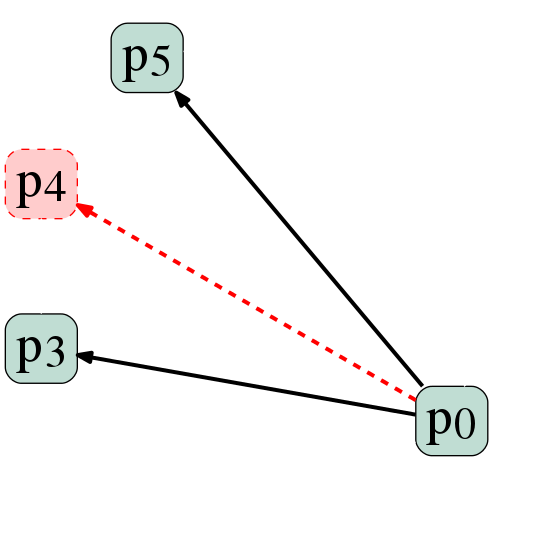}
}
\caption{(a) An overlay network connecting nine servers; (b) a digraph that satisfies $\text{I}_1$, $\text{I}_2$ and $\text{I}_3$; and (c) the digraph used by $p_6$ to track $m_0$.
Both digraphs consider nine servers connected through a $G_S(9,3)$ digraph~\cite{Soneoka:1996:DDC:227095.227101}
and are based on two failure notifications received by $p_6$ (indicated by dashed edges), 
one from $p_4$ indicating $p_0$'s failure and another from $p_7$ indicating $p_2$'s failure. 
Note that $p_4 \notin g[p_6][p_0]$.}
\label{fig:td_gs}
\end{figure}

$\text{I}_1$, $\text{I}_2$ and $\text{I}_3$ are necessary but not sufficient for a non-empty digraph 
to be a tracking digraph. 
As an example, we consider nine servers connected through a $G_S(9,3)$ digraph~\cite{Soneoka:1996:DDC:227095.227101} (see Figure~\ref{fig:gs_n9_d3_again}).
While $p_6$ is tracking $m_0$, it receives two notification, 
one from $p_4$ indicating $p_0$'s failure and another from $p_7$ indicating $p_2$'s failure; 
i.e., $F[p_6][p_0]=\{p_4\}$ and $F[p_6][p_2]=\{p_7\}$. Clearly, the digraph $\bar{K}_9(p_6)$ illustrated in Figure~\ref{fig:k9_p6}
satisfies the first three invariants. 
Yet, there is not reason for $p_6$ to suspect that $p_2$ has $m_0$. 

For sufficiency, invariant $\text{I}_4$ is needed. 
Together with $\text{I}_3$, $\text{I}_4$ requires that $p$ suspects only those servers that are connected 
through failures to the root $p_\ast$. In other words, for $p$ to suspect a server $q$, there must be a 
sequence of servers starting with $p_\ast$ and ending with $q$ such that every server preceding $q$ 
is both known to have failed and suspected to have sent $m_\ast$ to the subsequent server.
Note that a server is always connected through failures to itself.
Figure~\ref{fig:g_p6_p0} shows the actual digraph used by $p_6$ to track $m_0$ in the above example. 

\textbf{Second approach---TLAPS specification.} 
We use the above invariants to provide a non-recursive specification of a non-empty tracking digraph.
Let $K_n$ be a complete digraph with $n$ nodes; clearly, $K_n$ satisfies $\text{I}_1$ and $\text{I}_2$.
Let $\bar{K}_n(p)$ be a digraph obtained from $K_n$ by removing all edges not  
satisfying~$\text{I}_3$ (see Figure~\ref{fig:k9_p6} for nine servers connected through a $G_S(9,3)$ digraph).
Then, the set $\mathtt{TD}(p, p_\ast).\mathit{node}$ contains any node in $\bar{K}_n(p)$ that is either $p_\ast$
or (according to $p$) is connected to $p_\ast$ through failures, i.e.,
\begin{align}
\label{eq:td_nodes}
\{q \in \bar{K}_n(p).\mathit{node} \,:\, 
      &\vee q = p_\ast \\ 
      &\vee \exists \pi_{p_{\ast},q}(\bar{K}_n(p)) \,:\, 
            F[p][x] \neq \emptyset,\, \forall x \in \pi_{p_\ast,q}(\bar{K}_n(p)) \setminus \{q\} \},\nonumber
\end{align}
where $\pi_{p_\ast,q}(\bar{K}_n(p))$ is a path in $\bar{K}_n(p)$ from $p_\ast$ to $q$.
Note that when removing a node from $\bar{K}_n(p)$ we also remove all the edges incident on that node. 
Using TLAPS, we prove that this specification satisfies all four invariants~\cite{allconcur_tla}.

\subsection{The networking module}
\label{sec:net_spec}

The NET module specifies an interface for asynchronous message-based communication; 
the module assumes that servers communicate through an overlay network. 
The interface considers three constants: (1) $\mathit{S}$, the set of servers;
(2) $G$, the digraph that describes the overlay network; 
and (3) $\mathit{Message}$, the set of existing messages. 
We assume that every message has a field \emph{o} indicating the server that first sent it. 

Let $p \in \mathit{S}$ be any server. 
Then, the state of the NET module is described by the values of three variables (see Figure~\ref{fig:design}):
(1) $\mathit{sendBuf}[p]$, is $p$'s sending buffer;
(2) $\mathit{recvBuf}[p]$, is $p$'s receiving buffer; 
and (3) $\mathit{recvMsg}$, is the latest received message.
Note that while $\mathit{recvBuf}[p]$ is a sequence of received messages, 
$\mathit{sendBuf}[p]$ is a sequence of tuples, with each tuple consisting 
of a message and a sequence of destination servers. 
Also, both buffers act as FIFO queues. 
In the initial state, the buffers are empty sequences; the initial value of $\mathit{recvMsg}$ is irrelevant.

\textbf{Next-state relations.}
The NET module defines three operators that specify all the possible state 
transitions (see Table~\ref{tab:state_changes}).
The operators consists of the three main actions performed in message-based 
communication---sending, transmitting and delivering a message. 
To describe the three operators, let $\mathit{msgs}$ be a sequence of messages 
and $\mathit{nf}$ a mapping $\mathit{S} \rightarrow \{0,1\}$ indicating the non-faulty servers.

The $\mathtt{SendMsg}(p,\,\mathit{msgs},\,\mathit{nf})$ operator, 
updates $\mathit{sendBuf}[p]$ by appending tuples consisting of messages from $\mathit{msgs}$ with their destinations; 
for every message $m$, the set of destinations consist of $p$'s non-faulty successors, except for $m.o$. 
Note that the $\mathtt{SendMsg}$ operator has no precondition.

The $\mathtt{TXMsg}(p)$ operator sends $m$, the next message from $\mathit{sendBuf}[p]$, to $q$, one of $m$'s destinations; 
$m$'s sequence of destinations is updated by removing $q$; when there are no more destinations, 
$m$ is removed from $\mathit{sendBuf}[p]$. Also, $m$ is appended to $\mathit{recvBuf}[q]$. 
As a precondition, the send buffer of $p$ must not be empty.

The $\mathtt{DeliverMsg}(p)$ operator updates $\mathit{recvBuf}[p]$ by removing a message (i.e., the least-recent received)
and storing it in $\mathit{recvMsg}$. Note that $\mathit{recvMsg}$ is only a temporary variable used by the AB module to access 
the delivered message (i.e., the $\mathtt{ReceiveMessage}$ operator). 
As a precondition, the receive buffer of $p$ must not be empty.

\subsection{The failure detector module}
\label{sec:fd_spec}

The FD module provides information about faulty servers. It specifies a FD that guarantees both 
completeness and accuracy, i.e., a perfect FD~\cite{Chandra:1996:UFD:226643.226647}.
The specification assumes a heartbeat-based FD with local detection: 
Every server sends heartbeats to its successors; once it fails, its successors detect the lack of heartbeats. 
The module considers two constants: (1) $\mathit{S}$, the set of servers; 
and (2) $G$, the digraph that describes the overlay network. 

Let $p \in \mathit{S}$ be any server. 
Then, the state of the FD module is described by the values of one variable 
(see Figure~\ref{fig:design})---$\mathit{FD}[p]$, is an array of $n$ flags, one per server; 
$\mathit{FD}[p][q]$ indicates whether $p$ has detected $q$'s failure. Clearly, $\mathit{FD}[p][q]=1 \Rightarrow q \in p^+(G)$. 
Initially, all flags in $\mathit{FD}[p]$ are cleared. 

\textbf{Next-state relation.}
The FD module defines only one operator, $\mathtt{DetectFail}(p,q)$, 
that specifies the state transition when $p$ detects $q$'s failure; 
i.e., the $\mathit{FD}[p][q]$ flag is set (see Table~\ref{tab:state_changes}).
The operator has a set of preconditions. First, $p$ must be both non-faulty and a successor of $q$. 
Second, $q$ must be faulty, i.e., $\mathit{nf}[q]=0$; this condition 
guarantees the accuracy property required by a perfect FD~\cite{Chandra:1996:UFD:226643.226647}.

Once a failure is detected, the AB module must be informed. 
The FD module invokes the NET module to append a notification of $q$'s failure to $\mathit{recvBuf}[p]$ (see Figure~\ref{fig:design}).
This ensures that any A-broadcast messages sent by $q$ to $p$ that were already transmitted (i.e., added to $\mathit{recvBuf}[p]$) 
are delivered by $p$ before its own notification of $q$'s failure. 
Thus, Proposition~\ref{prop:in_order_trans} holds: 
If $p$ receives from $q_s \in q^+(G)$ a notification of $q$'s failure, 
then $q_s$ has not received from $q$ any message that $p$ did not already receive. In the above scenario,~$q_s = p$.

\subsection{Safety and liveness properties}
\label{sec:safe_and_live}

Using the above specification, we define both safety and liveness properties. 
First, \AC{} relies on a perfect FD for detecting faulty servers; 
hence, it guarantees both accuracy and completeness~\cite{Chandra:1996:UFD:226643.226647}. 
\emph{Accuracy} is a safety property:
It requires that no server is suspected to have failed before actually failing, i.e., 
\begin{align*}
 \forall q \in \mathit{S} : \mathit{nf}[q] = 1 \Rightarrow \forall p \in \mathit{S} \,:\, \mathit{FD}[p][q]=0.
\end{align*}
\emph{Completeness} is a liveness property: It requires that all failures are eventually detected, i.e., 
\begin{align*}
 \forall q \in \mathit{S} : \mathit{nf}[q] = 0 \rightsquigarrow 
        (\forall p \in q^+(G) \,:\, \mathit{nf}[p] = 1  \Rightarrow \mathit{FD}[p][q]=1),
\end{align*}
where $X \rightsquigarrow Y$ asserts that whenever $X$ is true, $Y$ is eventually true~\cite{Lamport:2002:SST:579617}.

Second, any atomic broadcast algorithm must satisfy four properties---validity, 
agreement, integrity, and total order~\cite{Hadzilacos:1994:MAF:866693, Chandra:1996:UFD:226643.226647}. 
Integrity and total order are safety properties. 
\emph{Integrity} requires for any message $m$, every non-faulty server to A-deliver $m$ at most once, 
and only if $m$ was previously A-broadcast by its owner $q$, i.e., 
\begin{align*}
 \forall p \in \mathit{S} : \mathit{nf}[p] = 1 \Rightarrow \forall q \in M[p] \,:\, \mathit{ab}[q] = 1.
\end{align*}
Note that the requirement that a server A-delivers $m$ at most once is ensured by construction, i.e., $M[p]$ is a set.

\emph{Total order} asserts that if two non-faulty servers $p$ and $q$ A-deliver messages $m_1$ and $m_2$, 
then $p$ A-delivers $m_1$ before $m_2$, if and only if $q$ A-delivers $m_1$ before $m_2$. 
Since $p$ A-delivered messages in $M[p]$ in a deterministic order, we replace total order with \emph{set agreement}:
Let $p$ and $q$ be any two non-faulty servers, then, after termination, $M[p] = M[q]$, i.e., 
\begin{align*}
 \forall p,q \in \mathit{S} : (\mathit{nf}[p] = 1 \wedge \mathit{done}[p] = 1 \wedge \mathit{nf}[q] = 1 \wedge \mathit{done}[q] = 1) \Rightarrow M[p] = M[q].
\end{align*}

Validity and agreement are liveness properties. 
\emph{Validity} asserts that if a non-faulty server A-broadcasts a message, then it eventually A-delivers it, i.e., 
\begin{align*}
&\forall p \in \mathit{S} : ( \Box(\mathit{nf}[p] = 1) \wedge \mathit{ab}[p] = 1) \rightsquigarrow \mathit{a\text{-}deliver}(p,p),
\end{align*}
where $\Box X$ asserts that $X$ is always true~\cite{Lamport:2002:SST:579617};
also, $\mathit{a\text{-}deliver}(p,q)=q \in M[p] \wedge \mathit{done}[p] = 1$ 
asserts the conditions necessary for $p$ to A-deliver the message A-broadcast by $q$.

\emph{Agreement} asserts that if a non-faulty server A-delivers a message A-broadcast by any server, 
then all non-faulty servers eventually also A-deliver the message. 
\begin{align*}
\forall p,q,s \in \mathit{S} : &( \Box(\mathit{nf}[p] = 1) \wedge \mathit{a\text{-}deliver}(p,s)) \rightsquigarrow (\mathit{nf}[q] = 1 \Rightarrow \mathit{a\text{-}deliver}(q,s)).
\end{align*}

To verify that all the above properties hold, we use the TLA+ Model Checker~\cite{Lamport:2002:SST:579617}, 
hence, the need for a tracking digraph specification that does not enumerate all paths of a digraph~(\cref{sec:td_update}).
For a small number of servers, e.g., $n=3$, the model checker can do an exhaustive search of all reachable states. 
Yet, for larger values of $n$ the exhaustive search becomes intractable. As an alternative, we use the model checker 
to randomly generate state sequences that satisfy both the initial state and the next-state relations. 
In the model, we consider the overlay network is described by a $G_S(n,d)$ digraph~\cite{Soneoka:1996:DDC:227095.227101}.
When choosing $G_S(n,d)$'s degree, i.e., its fault tolerance~(\cref{sec:allconcur}), 
we require a reliability target of 6-nines; the reliability is estimated over a period of $24$ hours 
according to the data from the TSUBAME2.5 system failure history~\cite{Sato:2012:DMN:2388996.2389022,tsubame}, 
i.e., server $\mathit{MTTF}\approx2$ years.

In addition, we use TLAPS to formally prove the safety properties---the 
FD's accuracy and the atomic broadcast's integrity and set agreement~(\cref{sec:proof}).
TLAPS does not allow for liveness proofs. However, validity and agreement require \AC{} to terminate; 
termination is informally proven in the \AC{} paper~\cite{poke2017allconcur}.

\section{\AC{}: formal proof of safety}
\label{sec:proof}

Atomic broadcast has two safety properties---integrity and total order. In \AC{}, total order can be replaced 
by set agreement~(\cref{sec:safe_and_live}). In addition, \AC{} relies on a perfect FD for information about 
faulty servers; thus, for safety, accuracy must also hold. 
For all three safety properties, we use TLAPS~\cite{Lamport:2002:SST:579617} to provide mechanically verifiable proofs~\cite{allconcur_tla}.
All three proofs follow the same pattern: We consider each property to be an invariant that holds for the initial state 
and is preserved by the next-state relations. 

\textbf{Accuracy.}
The FD's accuracy is straightforward to prove. Initially, all flags in $\mathit{FD}$ are cleared; hence, the property holds. 
Then, according to the specification of the $\mathtt{DetectFail}$ operator~(\cref{sec:fd_spec}), 
setting the flag $\mathit{FD}[p][q]$ for $\forall p,q \in \mathit{S}$ is preconditioned by $q$ previously 
failing, i.e., $\mathit{nf}[q]=0$. Moreover, due to the fail-stop assumption, a faulty 
server cannot become subsequently non-faulty. As a result, accuracy is preserved by the next-state relations.

\textbf{Integrity.}
Integrity is also straightforward to prove.  Initially, all the sets in $M$ are empty; hence, the property holds.
Then, the only two operators that update $M$ are $\mathtt{Abcast}$ and $\mathtt{RecvBCAST}$ (see Table~\ref{tab:state_changes}). 
The $\mathtt{Abcast}(p)$ operator adds $p$ to $M[p]$ and also sets the $\mathit{ab}[p]$ flag; hence, integrity 
is preserved by $\mathtt{Abcast}$. The $\mathtt{RecvBCAST}(p, m)$ operator adds $m.o$ to $M[p]$. 
Clearly, receiving a message $m$ entails the existence of a path from $m.o$ to $p$ such that each server on 
the path has $m$ in its $M$ set (see Lemma~\ref{lemma:MessagePath}); this also includes $m.o$.
When $m.o$ adds its message $m$ to $M[m.o]$, it also sets the $\mathit{ab}[m.o]$ flag (according to the $\mathtt{Abcast}$ operator). 
Thus, integrity is preserved also by  $\mathtt{RecvBCAST}$. 
In practice, to simplify the TLAPS proof, we introduce an additional precondition 
for the $\mathtt{RecvBCAST}(p, m)$ operator---$\mathit{ab}[m.o]=1$~(\cref{sec:ab_spec}).

\subsection{Set agreement}
The set agreement property is the essence of \AC{}---it guarantees the total order of broadcast messages. 
Clearly, set agreement holds in the initial state, since all $\mathit{done}$ flags are cleared. 
Moreover, the $\mathit{done}$ flags are set only by the $\mathtt{Adeliver}$ 
operator (see Table~\ref{tab:state_changes}); 
thus, we only need to prove set agreement is preserved by $\mathtt{Adeliver}$.
We follow the informal proof provided in the \AC{} paper~\cite{poke2017allconcur}. 
We introduce the following lemmas:

\begin{lemma}
\label{lemma:MessagePath}
Let $p$ be a server that receives $p_\ast$'s message $m_\ast$; then, there is a path (in G) from $p_\ast$ to $p$ 
such that every server on the path has received $m_\ast$ from its predecessor on the path, i.e., 
\begin{subequations}
\begin{align}
&\forall p, p_\ast \in \mathit{S} \,:\, p_\ast \in M[p] \Rightarrow \exists \pi_{p_\ast,p}(G)=\left(a_{1},\ldots,a_{\lambda}\right) \,:\, \nonumber\\
\label{eq:message_path_eq1}
&\quad\quad \wedge \forall k \in \{1,\ldots,\lambda\} \,:\, p_\ast \in M[a_k] \\
\label{eq:message_path_eq2}
&\quad\quad \wedge \forall q \in \mathit{S} : (\exists  k \in \{1,\ldots,(\lambda-1)\} : a_{k+1} \in F[q][a_k]) \Rightarrow p_\ast \in M[q].
\end{align}
\end{subequations}
\end{lemma}
\begin{proof}
Equation~\eqref{eq:message_path_eq1} is straightforward. 
Equation~\eqref{eq:message_path_eq2} ensures that every server on the path (except $p_\ast$) received $m_\ast$ from its predecessor. 
Any server $q$ that received a failure notification from a server on the path (except $p_\ast$) targeting its predecessor 
on the path also received $m_\ast$ (cf. Proposition~\ref{prop:in_order_trans}).  
\end{proof}

\begin{lemma}
\label{lemma:noDisconnect}
Let $p$ be a non-faulty server that does not receive $p_\ast$'s message $m_\ast$; 
let $\pi_{p_\ast,a_i}(G) = \left(a_{1},\ldots,a_{i}\right)$ be a path (in $G$)
on which $a_i$ receives $m_\ast$; 
let $\left(a_{1},\ldots,a_{i}\right)$ be also a path in $g[p][p_\ast]$. 
Then, $\mathit{done}[p] = 1 \Rightarrow (\forall q \in g[p][p_\ast].\mathit{node} \,:\, F[p][q] \neq \emptyset)$.
\end{lemma}
\begin{proof}
A necessary condition for $p$ to terminate is to remove every server 
$a_j,\,\forall 1 \leq j \leq i$ from $g[p][p_\ast]$.
According to \AC{}'s specification~(\cref{sec:td_update}), 
a server $a_j$ can be removed from $g[p][p_\ast]$ in one of the following scenarios:
(1) $p_\ast \in M[p]$; (2) $\nexists \pi_{p_\ast,a_j}(g[p][p_\ast])$; 
and (3) $\forall q \in g[p][p_\ast].\mathit{node} \,:\, F[p][q] \neq \emptyset$.
Clearly, $p_\ast \notin M[p]$. Also, $\nexists \pi_{p_\ast,a_j}(g[p][p_\ast])$ 
entails at least the removal of an edge from the $\left(a_{1},\ldots,a_{j}\right)$ path.
Let $(a_{l},a_{l+1}), 1 \leq l < j$ be one of the removed edges. Then, $a_{l+1} \in F[p][a_{l}]$, 
which entails $p_\ast \in M[p]$ (cf.~Lemma~\ref{lemma:MessagePath}). 
Thus, $p$ terminates only if $\forall q \in g[p][p_\ast].\mathit{node} \,:\, F[p][q] \neq \emptyset$.
\end{proof}

\begin{lemma}
\label{lemma:ak_in_g}
Let $p,q$ be two non-faulty servers such that $p_\ast \in M[p]$, but $p_\ast \notin M[q]$;
let $\pi_{p_\ast,p}(G) = \left(a_{1},\ldots,a_{\lambda}\right)$ be the path on which $p$ receives $m_\ast$; 
let $a_k$ be a server on $\pi_{p_\ast,p}(G)$ such that 
$\mathit{nf}[a_k] = 1$ and $\mathit{nf}[a_i] = 0,\,\forall 1 \leq i < k$.
Then, $\left(a_{1},\ldots,a_{i}\right),\,\forall 1 \leq i \leq k$ is a path in $g[q][p_\ast]$, 
i.e., $\left(a_{1},\ldots,a_{i}\right) \in g[q][p_\ast]$.
\end{lemma}
\begin{proof}
We use mathematical induction: 
The basic case is given by invariant $\text{I}_1$~(\cref{sec:td_update}), $a_1 = p_\ast \in g[q][p_\ast].\mathit{node}$.
For the inductive step, we assume $\left(a_{1},\ldots,a_{i}\right) \in g[q][p_\ast]$ for some $1 \leq i < k$.
Due to the FD's completeness property~(\cref{sec:safe_and_live}), the failure of $a_{j},\,\forall1\leq j \leq i$ is eventually detected;
due to the vertex-connectivity of $G$, $q$ eventually receives the failure notification of every $a_j$.
Moreover, $q$ cannot remove $a_j,\,\forall1\leq j \leq i$ from $g[q][p_\ast]$ before it receive failure notifications 
of every $a_j$ (cf.~Lemma~\ref{lemma:noDisconnect}). Thus, eventually $F[q][a_i] \neq \emptyset$ 
and, since $a_{i+1} \notin F[q][a_{i}]$ (cf.~Lemma~\ref{lemma:MessagePath}), $a_{i+1} \in g[q][p_\ast].\mathit{node}$ 
and $(a_i,a_{i+1}) \in g[q][p_\ast].\mathit{edge}$ (according to invariants~$\text{I}_2$ and~$\text{I}_3$).
\end{proof} 

\begin{figure}[!tp]
% \captionsetup[subfigure]{justification=centering}
\centering
% \subcaptionbox{$G_S(9,3)$\label{fig:gs_n9_d3}} {
\input{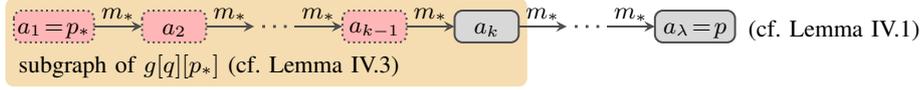}
% }
% \qquad
% \subcaptionbox{\label{fig:tracking_digraph}} {
% \scalebox{0.9}{\input{figures/tracking_digraph_small.tex}}
% }
\caption{The path $\pi_{p_\ast,p}$ on which a non-faulty server $p$ first receives $m_\ast$ from $p_\ast$; 
every server on the path (except $p_\ast$) has received $m_\ast$ from its predecessor.
Dotted red boxes indicate faulty servers; in case of no failures, $a_k=p_\ast$. Solid gray boxes indicate non-faulty servers. 
If another non-faulty server $q$ does not receive $m_\ast$, then it suspects all servers on the subpath from $a_1$ to $a_k$ to have $m_\ast$ and 
to have received $m_\ast$ from their predecessors on the path (except $p_\ast$).}
\label{fig:set_agreement_proof}
\end{figure}

\begin{theorem}
\label{th:set_agreement}
\AC{}'s specification guarantees set agreement. 
\end{theorem}
\begin{proof}
Let $p_\ast \in M[p]$, but $p_\ast \notin M[q]$. 
Then, $\exists \pi_{p_\ast,p} = \left(a_{1},\ldots,a_{\lambda}\right)$ on which $m_\ast$ first arrives at $p$ (cf.~Lemma~\ref{lemma:MessagePath}).
Let $a_k$ be a server on $\pi_{p_\ast,p}$ such that 
$\mathit{nf}[a_k] = 1$ and $\mathit{nf}[a_i] = 0,\,\forall 1 \leq i < k$;
%Let $a_k$ be the first non-faulty server on the path (closest to $p_\ast$); 
the existence of $a_k$ is given by the existence of $p$, a server that is both non-faulty and on $\pi_{p_\ast,p}$.
The path $\pi_{p_\ast,p}$ is illustrated in Figure~\ref{fig:set_agreement_proof}; the faulty servers are indicated by dotted red boxes.
Before $q$ terminates, $a_k \in g[q][p_\ast].\mathit{node}$ (cf.~Lemma~\ref{lemma:ak_in_g}).
However, due to the precondition of $\mathtt{Adeliver}(q)$, 
i.e., $\mathit{done}[q] = 1 \Rightarrow g[q][s].\mathit{node} = \emptyset,\,\forall s \in \mathit{S}$~(\cref{sec:ab_spec}), 
$a_k$ is subsequently removed from $g[q][p_\ast]$. 
Since $m_\ast \notin M[q]$ and $\mathit{nf}[a_k] = 1$, $a_k$ must have been disconnected from $p_\ast$~(\cref{sec:td_update}).
Thus, $q$ removed at least one edge from $\left(a_{1},\ldots,a_{k}\right)$; let $(a_{i},a_{i+1}), 1 \leq i < k$ be one of the removed edges. 
Then, $a_{i+1} \in F[q][a_{i}]$ (cf. $\text{I}_3$), which entails $p_\ast \in M[q]$ (cf.~Lemma~\ref{lemma:MessagePath}).
\end{proof}

\subsubsection{Reconstructed tracking digraphs}
The set agreement proof relies on the following property of tracking digraphs~(\cref{sec:td_update}):
If a server $q$ was removed from $g[p][p_\ast]$, then one of the following is 
true---(1) $p_\ast \in M[p]$; 
(2) $F[p][q] \neq \emptyset$\footnote{$F[p][q] \neq \emptyset$ is necessary, but not sufficient to remove~$q$ from $g[p][p_\ast]$.}; 
and (3) $\nexists \pi_{p_\ast,q}(g[p][p_\ast])$.
Yet, when the $\mathtt{Adeliver}(p)$ operator is enabled, $g[p][p_\ast].\mathit{node} = \emptyset$;
checking the existence of a path $\pi_{p_\ast,q}(g[p][p_\ast])$ is not possible. 
Thus, we reconstruct $g[p][p_\ast]$ from the failure notifications received by $p$; 
we denote the resulted digraph by $\mathtt{RTD}(p,p_\ast)$.

\textbf{RTD nodes.}
Constructing the set $\mathtt{RTD}(p,p_\ast).\mathit{node}$ 
is similar to the TLAPS specification of updating tracking digraphs~(\cref{sec:td_update}). 
The difference is twofold. First, if $p$ receives from $p_o$ a notification of $p_t$'s failure, with $p_t \in \mathtt{RTD}(p,p_\ast)$, 
then $p$ adds all $p_t$'s successors to $\mathtt{RTD}(p,p_\ast)$---including~$p_o$.
Second, servers are never removed from $\mathtt{RTD}(p,p_\ast)$. 
Clearly, every server added at any point to $g[p][p_\ast]$ is also in $\mathtt{RTD}(p,p_\ast)$

\textbf{RTD edges.}
To construct the set $\mathtt{RTD}(p,p_\ast).\mathit{edge}$, we first connect (in $\mathtt{RTD}(p,p_\ast)$) 
each server known to have failed to its successors. Then, we remove all the edges on which we are certain 
$p_\ast$'s message $m_\ast$ was not transmitted. To identify these edges we use Proposition~\ref{prop:in_order_trans}: 
If $p$ received from $e_2$ a notification of $e_1$ failure without previously receiving $m_\ast$, then 
the edge $(e_1,e_2)$ is not part of $\mathtt{RTD}(p,p_\ast)$, i.e., 
\begin{align}
\label{eq:rtd_edges}
&\forall e_1,e_2 \in \mathtt{RTD}(p,p_\ast).\mathit{node} \,:\, 
e_2 \in F[p][e_1] \wedge p_\ast \notin M[p] \Rightarrow (e_1,e_2) \notin \mathtt{RTD}(p,p_\ast).\mathit{edge}. 
\end{align}
Some of the remaining edges in $\mathtt{RTD}(p,p_\ast)$ were still not used for transmitting $m_\ast$, 
i.e., the edges for which $e_2 \in F[p][e_1]$ occurred before $p_\ast \in M[p]$. 
Since the $\mathtt{Adeliver}(p)$ operator has no access to the history of state updates modifying either $F[p][e_1]$ or $M[p]$, 
we cannot identify these edges. Yet, since $p_\ast \in M[p]$, these edges play no role in proving set agreement. 
Clearly, every edge added at any point to $g[p][p_\ast]$ is also in $\mathtt{RTD}(p,p_\ast)$.

\textbf{RTD invariant.} 
Using reconstructed tracking digraphs, we redefine the above property of tracking digraphs 
as an invariant, referred to as the \emph{RTD invariant}: 
\begin{align*}
&\forall p,q,p_\ast \in \mathit{S} \,:\, (q \in \mathtt{RTD}(p,p_\ast).\mathit{node} \wedge q \notin g[p][p_\ast].\mathit{node}) 
 \Rightarrow  p_\ast \in M[p] \vee  F[p][q] \neq \emptyset  \vee \nexists \pi_{p_\ast,q}(\mathtt{RTD}(p,p_\ast)). 
\end{align*}

The RTD invariant enables us to formally prove set agreement using TLAPS~\cite{allconcur_tla}. 
Clearly, when $\mathtt{Adeliver}(q)$ is enabled, $a_k$ (from the proof of Theorem~\ref{th:set_agreement}) 
is in $\mathtt{RTD}(q,p_\ast)$, but not in $g[q][p_\ast]$.
According to the initial assumptions, $p_\ast \notin M[q]$ and $F[q][a_k] = \emptyset$. 
Moreover, $\pi_{p_\ast,a_k} = \left(a_{1},\ldots,a_{k}\right)$ was at some point a path in $g[q][p_\ast]$ (cf.~Lemma~\ref{lemma:ak_in_g});
hence, $\pi_{p_\ast,a_k}$ is also a path in $\mathtt{RTD}(p,p_\ast)$ (follows from $\mathtt{RTD}$'s construction),  
which contradicts the RTD invariant.

\subsubsection{Proving the RTD invariant}
To prove the RTD invariant, we follow the same pattern as before: 
We prove that the invariant holds for the initial state (since $\mathtt{RTD}(p,p_\ast) = g[p][p_\ast],\,\forall p,p_\ast$) 
and is preserved by the only three operators that update the $M$, $G$ or $F$ variables---$\mathtt{Abcast}$; 
$\mathtt{RecvBCAST}$; $\mathtt{RecvFAIL}$ (see Table~\ref{tab:state_changes}).
In the following proofs, $X^\prime$ denotes the updated value of a variable $X$ after applying an operator~\cite{Lamport:2002:SST:579617}. 

\begin{theorem}
\label{th:rtd_inv_bcast}
Both $\mathtt{Abcast}$ and $\mathtt{RecvBCAST}$ operators preserve the RTD invariant. 
\end{theorem}
\begin{proof}
According to the specifications of both $\mathtt{Abcast}$ and $\mathtt{RecvBCAST}$,
$g[p][p_\ast] \neq g^\prime[p][p_\ast] \Rightarrow p_\ast \in M^\prime[p]$~(\cref{sec:ab_spec});
thus, we assume $g[p][p_\ast] = g^\prime[p][p_\ast]$. 
In addition, since $F[p]= F^\prime[p]$, $\mathtt{RTD}(p,p_\ast).\mathit{node}=\mathtt{RTD}(p,p_\ast)^\prime.\mathit{node}$.
Thus, since $M[p] \subseteq M^\prime[p]$, the only possibility for the RTD invariant to not be preserved by either operators 
is $\exists \pi_{p_\ast,q}(\mathtt{RTD}(p,p_\ast)^\prime)$. 
Let $(e_1,e_2)$ be one of the edges that enables such a path, 
i.e., $(e_1,e_2) \notin \mathtt{RTD}(p,p_\ast).\mathit{edge} \wedge (e_1,e_2) \in \mathtt{RTD}(p,p_\ast)^\prime.\mathit{edge}$.
As a result,  $e_2 \in F[p][e_1] \wedge p_\ast \notin M[p] \wedge (e2 \notin F^\prime[p][e_1] \vee p_\ast \in M^\prime[p])$ 
(cf.~Equation~\eqref{eq:rtd_edges}), which is equivalent to $p_\ast \in M^\prime[p]$.
\end{proof} 

\begin{theorem}
\label{th:rtd_inv_fail}
The $\mathtt{RecvFAIL}$ operator preserves the RTD invariant. 
\end{theorem}
\begin{proof}
We consider $\mathtt{RecvFAIL}(s,m)$, with $s \in \mathit{S}$; clearly, $M=M^\prime$ (see Table~\ref{tab:state_changes}). 
Also, if $s \neq p$, the RTD invariant is preserved, since also $F[p]=F^\prime[p]$ and $g[p]=g^\prime[p]$. 
Thus, the proof is concerned only with $\mathtt{RecvFAIL}(p,m)$. 
Aside from adding $m.o$ to $F[p][m.t]$, $F=F^\prime$. 
Clearly, $q=m.t \Rightarrow F^\prime[p][q] \neq \emptyset$.
Thus, we make the following assumptions, i.e., the only non-trivial case (cf.~RTD invariant):
\begin{align*}
  &\text{(A}_1\text{)}\quad q \neq m.t; & 
  &\text{(A}_3\text{)}\quad q \notin g^\prime[p][p_\ast].\mathit{node}; &
  &\text{(A}_5\text{)}\quad F^\prime[p][q] = \emptyset; \\
  &\text{(A}_2\text{)}\quad q \in \mathtt{RTD}(p,p_\ast)^\prime.\mathit{node}; &
  &\text{(A}_4\text{)}\quad p_\ast \notin M^\prime[p]; &
  &\text{(A}_6\text{)}\quad \exists \pi_{p_\ast,q}(\mathtt{RTD}(p,p_\ast)^\prime).
\end{align*}
We split the remainder of the proof in three steps. 

\emph{Step~1}. 
We show that $q$ is in $\mathtt{TD}(p,p_\ast)^\prime$~(\cref{sec:td_update}).
Let $\pi_{p_\ast,q}(\mathtt{RTD}(p,p_\ast)^\prime)=(a_1,\ldots,a_\lambda)$ (cf.~($\text{A}_6$)); 
then $(a_k,a_{k+1}) \in \mathtt{TD}(p,p_\ast)^\prime.\mathit{edge},\, \forall 1\leq k < \lambda$,
since $p_\ast \notin M^\prime[p]$ (cf.~($\text{A}_4$)) and no edges are removed from $\mathtt{RTD}(p,p_\ast)^\prime$ (cf.~Equation~\eqref{eq:rtd_edges}).
Thus, $a_k \in \mathtt{TD}(p,p_\ast)^\prime.\mathit{node}, \, \forall 1\leq k \leq \lambda$ (cf.~Equation~\eqref{eq:td_nodes}),
which entails the following assumption:
\begin{align*}
  &\text{(A}_7\text{)}\quad q \in \mathtt{TD}(p,p_\ast)^\prime.\mathit{node}.
\end{align*}

\emph{Step~2}. 
We show that $g[p][p_\ast]$ is not updated.
If $g^\prime[p][p_\ast] \neq g[p][p_\ast]$, then, $g^\prime[p][p_\ast]$ is either $\mathtt{TD}(p,p_\ast)^\prime$
or an empty digraph obtained after completely pruning $\mathtt{TD}(p,p_\ast)^\prime$. 
Since the latter option requires $F^\prime[p][s] \neq \emptyset,\, \forall s \in \mathtt{TD}(p,p_\ast)^\prime.\mathit{node}$~(\cref{sec:td_update}), 
which contradicts ($\text{A}_5$) and ($\text{A}_7$), we assume $g^\prime[p][p_\ast] = \mathtt{TD}(p,p_\ast)^\prime$.
Yet, this entails $q \in g^\prime[p][p_\ast].\mathit{node}$ (cf.~$\text{A}_7$), which contradicts~($\text{A}_3$). 
Thus, we make the following two equivalent assumptions~(\cref{sec:td_update}):
\begin{align*}
 &\text{(A}_8\text{)}\quad g[p][p_\ast]=g^\prime[p][p_\ast]; &
 &\text{(A}_9\text{)}\quad m.t \notin g[p][p_\ast].\mathit{node}.
\end{align*}

\emph{Step~3}. We distinguish between three scenarios---according to the specification~(\cref{sec:ab_spec}), 
$g[p][p_\ast]$ can be in one of the following three states: 
(1) an initial state, i.e., $g[p][p_\ast].\mathit{node}=\{p_\ast\}$, with $F[p][p_\ast] = \emptyset$; 
(2) a final state, i.e., $g[p][p_\ast].\mathit{node}=\emptyset$; 
and (3) an intermediary state, i.e., $g[p][p_\ast]=\mathtt{TD}(p,p_\ast)$.

The \emph{initial state} entails $g^\prime[p][p_\ast].\mathit{node}=\{p_\ast\}$ (cf.~($\text{A}_8$)); 
thus, $p_\ast$ is neither $m.t$ (cf.~($\text{A}_9$)) nor $q$ (cf.~($\text{A}_3$)).
Clearly, $p_\ast \neq m.t \Rightarrow F[p][p_\ast]=F^\prime[p][p_\ast]$.
Moreover, according to ($\text{A}_6$) and the construction of $\mathtt{RTD}(p,p_\ast)^\prime$, 
$p_\ast \neq q$ entails $F^\prime[p][p_\ast] \neq \emptyset$.
Yet, this contradicts the condition of the initial state, i.e., $F[p][p_\ast] = \emptyset$.

The \emph{final state}, i.e., $g[p][p_\ast].\mathit{node}=\emptyset$, 
can be reached only by completely pruning the tracking digraph (cf.~($\text{A}_4$) and $M[p] \subseteq M^\prime[p]$);
i.e., $F[p][s] \neq \emptyset,\, \forall s \in \mathtt{TD}(p,p_\ast).\mathit{node}$~(\cref{sec:td_update}).
According to invariant~$\text{I}_2$, no server can be added to $\mathtt{TD}(p,p_\ast)$, 
regardless of the received failure notification. 
As a result, $q \in \mathtt{TD}(p,p_\ast).\mathit{node}$ (cf.~($\text{A}_7$)) and, thus, $F[p][q] \neq \emptyset$.
Yet, since $F[p][q]=F^\prime[p][q]$ (cf.~($\text{A}_1$)), this contradicts ($\text{A}_5$).

The \emph{intermediary state} entails $g^\prime[p][p_\ast]=\mathtt{TD}(p,p_\ast)$ (cf.~($\text{A}_8$)); 
thus, $\nexists \pi_{p_\ast,q}(\mathtt{TD}(p,p_\ast))$ (cf.~($\text{A}_3$)).
Nonetheless, $\exists \pi_{p_\ast,q}(\mathtt{TD}(p,p_\ast)^\prime)= (b_1,\ldots,b_\gamma) \,:\, 
F^\prime[p][b_k] \neq \emptyset ,\, \forall 1\leq k < \gamma$  (cf.~($\text{A}_7$) and Equation~\eqref{eq:td_nodes}).
Moreover, $\exists 1 \leq k \leq \gamma : (b_k,b_{k+1}) \notin \mathtt{TD}(p,p_\ast).\mathit{edge}$;
let $k$ be the smallest index satisfying this property.
Then, either $F[p][b_k] = \emptyset$ or $b_{k+1} \in F[p][b_k]$ (cf.~invariant~$\text{I}_3$).
If $F[p][b_k] = \emptyset$, then $b_k = m.t$ (since $F^\prime[p][b_k] \neq \emptyset$).
Also, $(b_1,\ldots,b_k)$ is a path in $\mathtt{TD}(p,p_\ast)$ (cf.invariants~$\text{I}_1$, $\text{I}_2$ and~$\text{I}_3$) 
and, thus, $b_k \in \mathtt{TD}(p,p_\ast).\mathit{node}$, which entails $m.t \in g^\prime[p][p_\ast]$. 
Yet, this contradicts ($\text{A}_8$) and ($\text{A}_9$).
If $b_{k+1} \in F[p][b_k]$, then $b_{k+1} \in F^\prime[p][b_k]$ (i.e., $F[p][b_k] \subseteq F^\prime[p][b_k]$), 
which contradicts $(b_k,b_{k+1}) \notin \mathtt{TD}(p,p_\ast)^\prime.\mathit{edge}$.
%%%%%%%%%%%%%%%%%%%%%%%%%%%%%%%%%%%%%%%%%%%%%%

\end{proof}

%\section{Discussion}

%\textbf{What if \AC{} would use a complete digraph for communication?}

%\textbf{How would the assumption of eventual accuracy modify \AC{}'s specification?}

\section{Conclusion}

We have provided both a formal design specification and a formal proof of safety of \AC{}, a leaderless concurrent atomic broadcast algorithm. 
Previous work shows the advantage of using \AC{} over classic leader-based approaches, such as Paxos---it 
enables higher throughput for distributed agreement while being completely decentralized~\cite{poke2017allconcur}.
This work builds on \AC{} by both improving the understanding of the algorithm, through a TLA+ design specification, 
and formally proving the algorithm's safety property, using the TLA+ Proof System.

{\small
\textbf{Acknowledgements.}
This work was supported by the German Research Foundation (DFG) as part of the Cluster 
of Excellence in Simulation Technology (EXC 310/2) at the University of Stuttgart.
}

{%\small
\bibliographystyle{abbrv}
\bibliography{references}}

\end{document}